\documentclass[12pt,nopagetitle]{article}
\usepackage{amsfonts, amsmath, amssymb, amsthm, enumitem}  % различная математика
\usepackage[english]{babel}
\usepackage[utf8]{inputenc}
\usepackage[OT1]{fontenc}
\usepackage{bm}

\usepackage{textcomp, cmap, comment}	
\usepackage{graphicx, wrapfig}
\usepackage{epigraph, xcolor, hyperref}
\usepackage{array,tabularx,tabulary,booktabs}
\usepackage{euscript, mathrsfs}	

\newcounter{iii}

\newcommand{\R}{{\mathbb R}}

\theoremstyle{plain}
\newtheorem{thm}{Theorem}

\newtheorem{cla}{Claim}

\theoremstyle{definition}

\title{The VC-dimension of $k$-vertex $d$-polytopes}
\author{Andrey Kupavskii\footnote{MIPT, Moscow, IAS, Princeton and CNRS, Grenoble; Email: {\tt kupavskii@ya.ru}\ \  The author acknowledges the financial support from the Ministry of Education and Science of the Russian Federation in the framework of MegaGrant no 075-15-2019-1926.}}
\date{}

\begin{document}
\maketitle
\begin{abstract}
  In this short note, we show that the VC-dimension of the class of $k$-vertex polytopes in $\R^d$ is at most $8d^2k\log_2k$, answering an old question of Long and Warmuth.
\end{abstract}

Let $2^X$ stand for the power set of $X$. Let $\mathcal F\subset 2^{X}$ be a collection of subsets of a set $X$. % {\it range space} $(X, \mathcal R)$ is a set of points $X$ and a collection $\mathcal R$  of subsets of $X$.
In computational geometry, a pair $(X, \mathcal F)$ is called a {\it range space}. In learning theory, the sets from $\mathcal R$ are called {\it concepts} and $\mathcal R$ is called a {\it concept class}. We say that a set $Y\subset X$ is {\it shattered} by $\mathcal F$ if for any $Y'\subset Y$ there is $F'\in \mathcal F$ such that $F'\cap Y = F'$. The {\it VC-dimension} of a collection of sets is the size of the largest shattered subset $Y$ of $X$. It was defined in a seminal paper of Vapnik and Chervonenkis
\cite{VC}. The VC-dimension is one of the fundamental complexity characteristics of a collection of sets and has numerous applications across different branches of mathematics, such as combinatorics, graph theory, computational geometry, learning theory, model theory etc.

It is typical to consider set systems, induced on a set $X$ of all points in some space by all objects from a certain geometrically defined class. We shall omit $X$ from the notation when it is clear from the context. One crucial example is the class of all half-spaces in $\mathbb R^d$. The VC-dimension of this class is $d+1$, a statement that is essentially equivalent to the classical theorem of Radon in convex geometry (see \cite{Mat}). To extend this to polytopes, we need one general result of Blumer, Ehrenfeucht, Haussler and Warmuth
\cite{BEHW}.
\begin{thm}[\cite{BEHW}]
  Let $\mathcal R\subset 2^{X}$ be a collection of sets of VC-dimension $t$. Let $\mathcal R^{\cap k}$ be the collection of all sets that can be obtained as intersections of at most $k$ sets from $\mathcal R$. Then the VC-dimension of $\mathcal R^{\cap k}$ is $O(tk\log k)$.
\end{thm}

Using this theorem and Radon's theorem, we conclude that the VC-dimension of the class of all polytopes in $\mathbb R^d$ with at most $k$ faces is $O(dk\log k)$. In the paper \cite{CKM}, Csik\'os,  Mustafa and the author showed that this bound is tight up to a constant factor for any $d\ge 4$. We refer to \cite{CKM} for more background on the topic.

Surprisingly, very little is known about the dual (in the geometric sense) class of polytopes in $\mathbb R^d$ with at most $k$ vertices. The famous Upper Bound Theorem (see \cite{Mat}) implies that a polytope with $k$ vertices can have at most $O(k^{\lfloor d/2\rfloor})$ faces, which gives an upper  bound of roughly $d^2k^{d/2}$ on the VC-dimension of this class.
Long and Warmuth \cite{LW} asked 30 years ago whether the VC-dimension of this class is polynomial in $d$. In this note, we answer this question in the positive.

\begin{thm} The VC-dimension of the class of polytopes with at most $k$ vertices in $\mathbb R^d$ is at most $8d^2k\log_2 k$.
\end{thm}

\begin{proof}
Assume that $A\subset \R^d$, $A := \{a^1\ldots, a^t\}$, is shattered by polytopes with at most $k$ vertices. Let $x^i:=(x_1^i,\ldots,x_d^i)$, $i=1,\ldots, k$, be the variables that represent the vertices of any such polytope. For each $a^j\in A$ and a $(d+1)$-tuple $x^{i_1},\ldots, x^{i_{d+1}}$ with $i_1< \ldots< i_{d+1}$, we define $2(d+1)$ polynomials of the following form. For $s =1,\ldots, d+1$ put
\begin{align*} P_{j;i_1,\ldots,i_{d+1}}^{s,s}(x^1,\ldots,x^k) \ &:=\ \det[x^{i_1}-x^{i_s},\ldots, x^{i_{s-1}}-x^{i_s},x^{i_{s+1}}-x^{i_s},\ldots, x^{i_{d+1}}-x^{i_s}],\\
P_{j;i_1,\ldots,i_{d+1}}^{s,0}(x^1,\ldots,x^k) \ &:=\ \det[x^{i_1}-a^{j},\ldots, x^{i_{s-1}}-a^{j},x^{i_{s+1}}-a^{j},\ldots, x^{i_{d+1}}-a^{j}].
\end{align*}

Each of these polynomials has degree $d$ and actually depends only on at most $d+1$ variables. Note that, for fixed $x^{i_1},\ldots,x^{i_{d+1}}$, $s$ and $a^j$, the signs of these two polynomials simply indicate the orientation of the simplices on the vertices $(x^{i_1},
\ldots, x^{i_{d+1}})$ and $(x^{i_1},
\ldots,x^{i_{s-1}},a^j,x^{i_{s+1}},
\ldots, x^{i_{d+1}})$. If these orientations coincide then $x^{i_s}$ and $a^{j}$ lie on the same side from the hyperplane spanned by $x^{i_1},
\ldots, x^{i_{s-1}},x^{i_s},
\ldots, x^{i_{d+1}}$. This simple observation implies the following claim.
\begin{cla} The point $a^j$ is contained in the convex hull $\mathrm{conv}(y^{i_1},\ldots, y^{i_{d+1}})$ for some fixed $y^{i_1},\ldots, y^{i_{d+1}}$ iff for each $s=1,\ldots, d+1$ either the signs of $P^{s,s}_{j,i_1,\ldots, i_{d+1}}(x^1,\ldots, x^k)$ and $P^{s,0}_{j,i_1,\ldots, i_{d+1}}(x^1,\ldots, x^k)$ coincide or $P^{s,0}_{j,i_1,\ldots, i_{d+1}}(x^1,\ldots, x^k)=0$ at any point satisfying $x^{i_r}=y^{i_r}$, $r=1,\ldots, d+1$. %(We can substitute anything for the other variables.)
\end{cla}

The important conclusion derived from Claim 1 is that the containment of $a^j$ in ${\rm conv}(y^{i_1},\ldots, y^{i_{d+1}})$ is determined by the signs of the polynomials defined above.
In total, we have defined $(2d+2)\cdot t\cdot {k\choose d+1}$ polynomials of degree $d$ each and that depend on $kd$ coordinates.\footnote{Actually, we have a bit fewer different polynomials since, e.g., $P^{s,s}$ does not depend on $a^j$ and $P^{s,0}$ does not depend on $x^{i_s}$, but this is not important for our purposes.}

A vector $v\in \{0,\pm 1\}^l$ is called a {\it sign pattern} of a set $p_1,\ldots,p_l$ of polynomials if there exists a vector $z$ such that the sign of $p_i(z)$ is $v_i$ for all $i$. We use the following classical result from combinatorial algebraic geometry.

\begin{thm}[Oleinik--Petrovskii \cite{OP}, Milnor \cite{Mi}, Thom \cite{Th}] Let $p_1,\ldots, p_l$ be $m$-variate real polynomials of degree at most $D$. Then the number of different sign patterns that these polynomials have is at most $$\Big(50Dl/m\Big)^m.$$
\end{thm}

Using it for our situation, we get that the total number of sign patterns of the polynomials we defined is at most $$\Big(\frac{50d\cdot 2(d+1)t{k\choose d+1}}{kd}\Big)^{kd}<(100tk^{d})^{kd}<2^{(7+\log_2 t+d\log_2 k)kd}.$$

Let us relate sign patterns to the property of a point being contained in the polytope. By Caratheodory's theorem (see \cite{Mat}), a point $z$ is contained in a convex hull of a set $X$ of $k$ points if and only if there is a subset $Y\subset X$ of size $d+1$ such that $z\subset {\rm conv}(Y)$. Thus, a point $a^i$ is contained in ${\rm conv}(y^1,\ldots, y^k)$ if and only if there exists a set of indices $i_1,\ldots, i_{d+1}$ for which the conditions of Claim~1 are fulfilled. %In short, to determine whether each point $a^j$ is in the convex hull of a polytope on a given set of $k$ vertices, we apply a certain function $F$, which maps each sign pattern to $\{0,1\}$, to the where $F = 1$ iff $a_j$ is contained in the convex hull.

Crucially, the discussion in the previous paragraph implies that there exists a correspondence between sign patterns of our polynomials and subsets $A'\subset A$ that can be obtained by intersecting $A$ with a $k$-vertex polytope, moreover, each sign pattern corresponds to only one subset.\footnote{Different sign patterns may correspond to the same subset.} Thus, the total number of different subsets that we can obtain by intersecting $A$ with $k$-vertex polytopes is at most the number of sign patterns of our polynomials.
% corresponds to the number of possible vectors $v\in \{0,1\}^t$ we can obtain, where $v_j = F(a^j,x^1,\ldots,x^l)$ for each $j=1,\ldots, t$ with some fixed $x^1,\ldots, x^l$, is at most the number of sign patterns.  (The image cannot be bigger than the set on which the function is defined).
At the same time, in order to shatter $A$, we need to obtain all $2^t$ different subsets. Therefore, we have
$$2^t\le 2^{(7+\log_2 t+d\log_2 k)kd} \ \ \ \Leftrightarrow \ \ \ t\le (7+\log_2 t+d\log_2 k)kd.$$
The last inequality is violated for $t = 8d^2k\log_2 k$ and $d,k\ge 3$, which concludes the proof.
\end{proof}

It is natural to ask for the lower bounds. It is not difficult to show that the VC-dimension of $k$-vertex polytopes in $\mathbb R^d$ is at least $\frac 13 dk$ for $k\ge 2d\ge 4$. We describe one possible construction below. More precisely, we describe a construction of a set of $k(d-1)$ points in $\mathbb R^d$ that is shattered by polytopes with $k+d-1$ vertices. Let us sketch the construction. Take a unit circle and take $k$ points $x_1,\ldots, x_k$ on the circle. Replace each point $x_i$ by a small regular $(d-1)$-vertex simplex $S_i$ centered at $x_i$ and lying in the $(d-2)$-dimensional plane that is orthogonal to the circle and passes through $x_i$. Choose $S_i$ so that all of them are translates of each other. We claim that the set $A:=S_1\cup \ldots \cup S_k$ is shattered by $(k+d-1)$-vertex polytopes.

Indeed, take a $(d-2)$-plane orthogonal to the circle and passing through the center of the circle. Take $d-1$ points in this plane that form a huge regular $(d-1)$-simplex $S$ centered at the center of the circle. Choose $S$ so that it is a reflected copy of $S_i$ after translation and homothety (i.e., the centers of hyperfaces of $S_i$ form a homothetic translated copy of $S$). The set $S$ of $d-1$ vertices is common for all the polytopes that we shall use. Next, fix a subset $A'\subset A$ that we want to separate by a polytope. Put $S'_i:=A'\cap S_i$. Each $S'_i$ is a face of the simplex.  Consider the line $l_i$ passing through the center of the circle and the center of the face $S'_i$. For each $i$, take a point $y_i$ that lies on $l_i$ distance $\epsilon(|S_i'|)$ farther than the center of the face $S'_i$. It is not so difficult, although not straightforward, to see that, for appropriate choices of $\epsilon(|S'_i|)$ we have ${\rm conv} (S\cup y_i)\supset S'_i$ and that, moreover, $S_i\setminus S'_i\not\subset {\rm conv} (S\cup y_1\cup\ldots\cup y_k)$.

It is an exciting open problem to close the gap between the lower (roughly $dk$) and the upper (roughly $d^2k\log k$) bounds.
\\

{\sc Acknowledgements: } We thank Lee-Ad Gottlieb, Aryeh Kontorovich and Gabriel Nivasch for sharing the question and for several enlightening discussions.


\begin{thebibliography}{10}
\bibitem{BEHW}
A. Blumer, A. Ehrenfeucht, D. Haussler, and M. K. Warmuth.   {\it Learnability andthe Vapnik-Chervonenkis dimension}, J. ACM 36 (1989), N4, 929--965.


\bibitem{CKM} M. Csik\'os, A. Kupavskii, N. Mustafa, {\it Optimal bounds on the VC-dimension}, Journal of Machine Learning Research 20 (2019), 81.1--81.8.


\bibitem{LW} P.M. Long, M.K. Warmuth, {\it Composite Geometric Concepts and Polynomial Predictability}, COLT (1990),  273--287.

\bibitem{Mat} J. Matou\v sek, {\it Lectures on Discrete Geometry}, Springer, Vol. 212 (2002), New York.

\bibitem{Mi} J.W. Milnor, {\it On the Betti numbers of real algebraic varieties}, Proc. Amer. Math. Soc. 15 (1964), 275--280.

\bibitem{OP} O.A. Oleinik, I.B. Petrrovskii, {\it On the topology of real algebraic surfaces} (in Russian), izv. Akad. Nauk SSSR 13 (1949), 389--402.


\bibitem{Th} R. Thom, {\it On the homology of real algebraic varieties} (in French), In S.S. Cairns, editor, Differential and Combinatorial Topology. princeton Univ. Press, 1965.

\bibitem{VC} V. N. Vapnik and A. Ya. Chervonenkis, {\it On the uniform convergence of relative frequenciesof events to their probabilities}, Theory of Probability and its Applications 16 (1971), N2, 264--280.


\end{thebibliography}
\end{document}